\documentclass{article}

 \PassOptionsToPackage{numbers, compress}{natbib}
\usepackage[preprint]{neurips_2021}

\usepackage{color}

\usepackage{amsmath,amssymb}
\usepackage{ulem}
\usepackage{hyperref}
\hypersetup{
    citecolor=blue,
    colorlinks=true,
    linkcolor=blue,
    filecolor=magenta,      
    urlcolor=blue,
linktocpage}
\usepackage[nameinlink]{cleveref}

\usepackage[toc,page,header]{appendix}
\usepackage{minitoc}

\doparttoc 
\faketableofcontents 

\usepackage[vlined,ruled,linesnumbered]{algorithm2e} 
\SetKwRepeat{Do}{do}{while}

\usepackage{verbatim} 

\usepackage{amsthm} 
\usepackage{thmtools} 

\newtheorem{theorem}{Theorem}
\newtheorem{corollary}{Corollary}
\newtheorem{lemma}{Lemma}

\newtheorem{definition}{Definition}

\newcommand{\bbA}{\mathbb{A}}

\newcommand{\bbE}{\mathbb{E}}

\newcommand{\bbS}{\mathbb{S}}

\usepackage[utf8]{inputenc} 
\usepackage[T1]{fontenc}    
\usepackage{hyperref}       
\usepackage{url}            
\usepackage{booktabs}       
\usepackage{amsfonts}       
\usepackage{nicefrac}       
\usepackage{microtype}      
\usepackage{xcolor}         
 \usepackage{tcolorbox}
\title{
On the Complexity of Computing Markov Perfect Equilibrium in General-Sum Stochastic Games
}

\author{
	 Xiaotie Deng\thanks{Corresponding authors.} \\
     Center on Frontiers of Computing Studies\\
	 Peking University\\
	 \texttt{xiaotie@pku.edu.cn} \\
	 \And
	Ningyuan Li$^*$\\
	Center on Frontiers of Computing Studies\\
	Peking University\\
	\texttt{yuhaoli.cs@pku.edu.cn} \\
	 \And
	David Mguni\\
	Huawei R\&D UK\\
	\texttt{davidmguni@hotmail.com} \\
	 \And
	Jun Wang\\
	University College London\\
	\texttt{jun.wang@cs.ucl.ac.uk} \\
	 \And
	Yaodong Yang$^*$\\
	Institute for AI, Peking University\\
	\texttt{yaodong.yang@pku.edu.cn} \\
}

\begin{document}

\maketitle

\begin{abstract}
Similar to the role of Markov decision processes in reinforcement learning, Stochastic Games (SGs) lay the foundation for the study of multi-agent reinforcement learning (MARL) and sequential agent interactions. In this paper, we derive that computing an approximate Markov Perfect Equilibrium (MPE) in a finite-state discounted Stochastic Game within the exponential precision is \textbf{PPAD}-complete. We adopt a function with a polynomially bounded description in the strategy space to convert the MPE computation to a fixed-point problem, even though the stochastic game may demand an exponential number of pure strategies, in the number of states, for each agent. The completeness result follows the reduction of the fixed-point problem to {\sc End of the Line}. Our results indicate that finding an MPE in SGs is highly unlikely to be \textbf{NP}-hard unless \textbf{NP}=\textbf{co-NP}. Our work offers confidence for MARL research to study MPE computation on general-sum SGs and to develop fruitful algorithms as currently on zero-sum SGs. 

\end{abstract}

\section{Introduction}
The seminal work of Shapley~\cite{shapley1953stochastic} 
defines Stochastic Games (SGs) to study the dynamic non-cooperative multi-player game, where each player simultaneously and independently chooses an action at each round, and the next state is determined by a probability distribution depending on the current state and the chosen joint actions. 
In two-player zero-sum SGs, Shapley \cite{shapley1953stochastic}  proved the existence of a stationary strategy profile  in which no agent has an incentive to deviate; similarly,  the existence of equilibrium in stationary  strategies also holds in multi-player nonzero-sum SGs \cite{fink1964equilibrium}. 
Such a solution concept (now also known as \textsl{Markov perfect equilibrium} (MPE) \cite{maskin2001markov}) models the dynamic nature of  multi-player games. 
As a refinement of Nash equilibrium \cite{nash1951} on SGs,  MPE  prevents non-payoff-relevant variables from affecting strategic behaviors, which allows researchers to identify the impact of state variables on outcomes.
Due to its generality, 
the framework of SGs has  enlightened  a sequence of  studies \cite{neyman2003stochastic,solan2015stochastic}  on a wide range of real-world applications   ranging from advertising and pricing \cite{albright1979birth}, fisheries modelling  \cite{sobel1982stochastic}, football player selection \cite{winston1984stochastic}, travelling inspection \cite{filar1985player}, and designing modern gaming AIs \cite{peng2017multiagent}. 
As a result, developing algorithms to compute MPE in  SGs has become one of the key subjects in an extremely rich research domain, including but not limited to  applied mathematics, economics, operations research, computer science and artificial intelligence   \cite{filar2012competitive,raghavan1991algorithms,solan2015stochastic}.

SGs underpin many AI/machine learning studies. For example, it is the key framework for studying adversarial training \cite{franci2020game, goodfellow2014generative} and modelling robustness \cite{pinto2017robust,abdullah2019wasserstein} in zero-sum setting.  
In reinforcement learning (RL), 
SG extends the Markov decision process (MDP) formulation to incorporate strategic interactions. 
Similar to the role of MDP in RL \cite{sutton2018reinforcement},    
SGs build the foundation for  multi-agent reinforcement learning (MARL) techniques  to study optimal decision makings in multi-player games \cite{littman1994markov}. 
In  last decades, a wide variety of MARL algorithms have been developed to solved SGs \cite{yang2020overview}. 

Computing a MPE in (general-sum) SGs requires a perfect knowledge of the transition dynamics and the payoffs of the game \cite{filar2012competitive}, which is often infeasible in practice. 
To overcome this difficulty, MARL methods are often applied to learn the MPE of a SG based on the interactions between agents and the environment. 
MARL algorithms are generally considered under two settings: \textsl{online} and \textsl{offline}.  
In the offline setting (also known as the batch setting \cite{perolat2017learning}), the learning algorithm controls all players in a centralised way, with the hope that the learning dynamics can eventually lead to a MPE by using limited number of interaction samples. 
In the online setting, the learner controls only one of the players to play with arbitrary opponents in the game, assuming having unlimited access to the game environment,  and  the central focus is often about the \textsl{regret}: the difference between a benchmark measure (often in hindsight) and the learner's total reward during learning.  

In the offline setting, two-player zero-sum (discounted) SGs have been extensively studied. 
Since the opponent is purely adversarial in zero-sum SGs, the process of seeking for the worst-case optimality for each player can be thought of as solving MDPs. 
As a result, (approximate) dynamic programming methods \cite{bertsekas2008approximate,szepesvari1996generalized} such as LSPI \cite{lagoudakis2003least}  and FQI \cite{munos2008finite} / NFQI \cite{riedmiller2005neural} can be adopted to solve SGs \cite{perolat2016softened,lagoudakis2002value,perolat2015approximate,sidford2020solving,jia2019feature}. 
Under this setting, policy-based methods ~\cite{daskalakis2020independent, hansen2013strategy} can also be applied. 
However, directly applying exiting MDP solvers on  general-sum SGs are challenging. 
Since solving two-player NE in general-sum normal-form games (i.e., one-shot SGs) is well-known to be PPAD-complete \cite{daskalakis2009complexity,chen2006settling}, the complexity of MPE in general-sum SGs are expected to be at least PPAD. 
Although early attempts such as Nash-Q learning \cite{hu2003nash}, Correlated-Q learning \cite{greenwald2003correlated}, Friend-or-Foe Q-Learning \cite{littman2001friend} have been made to solve general-sum SGs under strong assumptions,  
Zinkevich et. al. \cite{zinkevich2006cyclic} demonstrated that the entire class of value iteration methods cannot find stationary NE policies in general-sum SGs.
The difficulties on both the complexity side and the algorithmic side lead to very few existing MARL solutions to  general-sum SGs; successful approaches either assumes knowing the complete information of the SG  and thus solving MPE can be turned into an optimisation problem \cite{prasad2015two}, or,   proves the convergence of batch RL methods to a weaker notion of NE  \cite{perolat2017learning}.

In the online setting, one of the most well-known algorithm is R-MAX \cite{brafman2002r}, which studied (average-reward) zero-sum SGs and provided a polynomial (in terms of game size and error parameter) regret bound when competing with an arbitrary opponent.   
Under the same regret definition, recently, UCSG \cite{wei2017online} improved R-MAX and achieved a sublinear regret, but still in two-player zero-sum SGs. 
When it comes to MARL solutions, Littman \cite{littman1994markov} proposed a practical solution named Minimax-Q that replaces the max operator with the minimax value. Asymptotic convergence results of Minimax-Q in both tabular cases \cite{littman1996generalized} and value function approximations \cite{fan2020theoretical} have been shown. 
Yet, playing the minimax value could be overly pessimistic. If the adversary plays sub-optimally, the learner could achieve a higher reward.
To account for this, WoLF  \cite{bowling2001rational} was proposed; and unlike Minimax-Q, WoLF is \emph{rationale} in the sense that it can exploit opponent's policy.     
AWESOME \cite{conitzer2007awesome}
 further generalised WoLF and achieve NE convergence in  multi-player general-sum repeated games. 
However, 
outside the scope of zero-sum SGs, 
the question \cite{brafman2002r} of  whether a polynomial time no-regret (near-optimal) RL/MARL  algorithm exists for general-sum SGs is still unanswered. 

Although SG has been proposed for more than 60 years and despite its importance, surprisingly, the complexity of finding a  MPE in  SG has never been answered. 
In fact, unlike the fruitful results on zero-sum SGs,    we still know very little about the complexity of solving general-sum SGs. 
Two relevant results we know   are   
that determining whether a pure-strategy NE exist in a  SG is \textbf{PSPACE}-hard \cite{conitzer2008new}, and it is \textbf{NP}-hard to determine if there exists a memoryless $\epsilon$-NE in \emph{reachability} SGs  \cite{chatterjee2004nash}. 
It is long projected solving MPE in (infinite-horizon)  SGs is at least \textbf{PPAD}-hard, since solving a two-player NE in one-shot SGs is already \textbf{PPAD}-hard  \cite{daskalakis2009complexity,chen2006settling}. 
This suggests that under computational hardness assumption, it is unlike to have polynomial-time algorithms in even two-player stochastic games. 
Yet, the unresolved question is that 

\begin{tcolorbox}[fonttitle=\normalsize,fontupper=\normalsize,fontlower=\normalsize,top=1pt,bottom=1pt,left=1pt,right=1pt,title=The key question that we try to address in this paper:]
\centering
\emph{Can solving MPE in general-sum SGs be anywhere harder in the complexity class?}
\end{tcolorbox}

In this paper, we answer to the above question negatively by proving that computing a MPE in a finite-state discounted SG is \textbf{PPAD}-complete. 
Based on our result, we given an affirmative answer that   finding an MPE in SGs is highly unlikely to be \textbf{NP}-hard under the circumstance that \textbf{NP}$\neq$ \textbf{co-NP}.  
We hope this result could encourage MARL researchers to work more on general-sum SGs, leading to  fruitful MARL solutions  as those currently on zero-sum SGs. 

\subsection{Intuitions and a Sketch of Our Main Ideas}
Like the classic complexity class \textbf{NP}, \textbf{PPAD} is a collection of computational problems. As the definition of \textbf{NP}-completeness, a problem is said to be \textbf{PPAD}-complete if it is in \textbf{PPAD}, and is at least as hard as every problem in \textbf{PPAD}. When one Stochastic Game has only one state and the discount factor $\gamma=0$, then finding a Markov perfect equilibrium (MPE) is equivalent to finding a Nash equilibrium in the corresponding normal-form game, which is known to be \textbf{PPAD}-complete \cite{daskalakis2009complexity,chen2006settling}. So the \textbf{PPAD}-hardness of finding MPE is relatively direct (\Cref{PPAD-hard}).

To obtain the \textbf{PPAD}-complete result (\Cref{PPAD-complete}), it is sufficient for us to prove the \textbf{PPAD} membership of MPE (\Cref{PPAD membership}). 

\textbf{i)} The first key observation is that we can construct a function $f$ of the strategy profile space, such that each strategy profile is a fixed point of $f$ if and only of it is an MPE (\Cref{theorem: MPE exists}). Further, we prove the function $f$ is continuous (actually $\lambda$-Lipschitz by \Cref{lemma: f Lipschitz}), so that fixed points are guaranteed to exist by the Brouwer fixed point theorem.

\textbf{ii)} We then prove the function $f$ has some ``good'' approximation properties. Let $|\mathcal{SG}|$ be the input size of a stochastic game. If we can find a $\texttt{poly}(|\mathcal{SG}|)\epsilon^2$-approximate fixed point $\pi$ of $f$, i.e., $\|f(\pi)-\pi\|_{\infty}\leq \texttt{poly}(|\mathcal{SG}|)\epsilon^2$, where $\pi$ is a strategy profile, then $\pi$ is an $\epsilon$-approximate MPE for the Stochastic Game (combining \Cref{lemma: single state} and \Cref{lemma: all states}). So our goal converts to finding an approximate fixed point.

\textbf{iii)} To prove the \textbf{PPAD} membership of finding an MPE, we will reduce it to the problem {\sc End of the Line} (whose formal definition is in \Cref{sec: PPAD and MPE problem}), which is the first \textbf{PPAD}-complete problem introduced by Papadimitriou~\cite{papadimitriou1994complexity}. We will show, the reduction could be constructed in polynomial time, and every solution of the problem {\sc End of the Line} corresponds to a good approximate fixed point (\Cref{lemma: stopping appro}), thus yields an $\epsilon$-approximate MPE.

\section{Stochastic Games}

\begin{definition}[Stochastic Game]\label{definition: Stochastic Game}
A Stochastic Game is defined by a tuple of key elements $\left\langle n,\bbS,\bbA,P,r,\gamma\right\rangle$, where 
\begin{itemize}
    \item $n$ is the number of agents.
    \item $\bbS$ is the set of finite environmental states. Suppose that $|\bbS|=S$.
    \item $\bbA=\bbA^1\times\cdots\times \bbA^n$ is the set of agents' joint actions. Suppose that $|\bbA^i|=A^i$ and $A_{\max}=\max_{i\in[n]} A^i$.
  \item $P: \bbS \times \bbA \rightarrow \Delta(\bbS)$ is the transition probability, that is, at each time step, given the agents' joint action $a\in \bbA$, then the transition probability from state $s$ to state in the next time step $s'$ is $P(s'|s,a)$.
    \item $r=r^1\times\cdots\times r^n: \bbS\times \bbA \rightarrow \mathcal{R}_+^n$ is the reward function, that is, when an agents are at state $s$ and play a joint action $a$, then the agent $i$ will get reward $r^i(s,a)$. We assume that the rewards are uniformly bounded by $R_{\max}$.
    \item $\gamma\in \left[0,1\right)$ is the discount factor that specifies the degree to which the agent’s rewards are discounted over time.
\end{itemize}
\end{definition}

Each agent aims to find a behavioral strategy with Markovian property, meaning that each agent's strategy can be conditioned only on the current state of the game. 

Note that behavioral strategy is different from mixed strategy. To be more clear, we give both definitions of mixed strategy and behavioral strategy.

The pure strategy space of an agent $i$ is $\prod_{s\in\bbS}\bbA^i$, meaning that the agent $i$ needs to select an action at each state. Note that the size of pure strategy space of each agent is $|\bbA^i|^S$, which is exponential in the number of states.
\begin{definition}[Mixed Strategy]
The mixed strategy space is $\Delta\left(\prod_{s\in\bbS}\bbA^i\right)$, i.e., the probability distribution on pure strategy space $\prod_{s\in\bbS}\bbA^i$.
\end{definition}

\begin{definition}[Behavioral Strategy]
A behavioral strategy of an agent $i$ is $\pi^i: \bbS\rightarrow \Delta(\bbA^i)$, i.e., $\forall s\in\bbS, \pi^i(s)$ is a probability distribution on $\bbA^i$.
\end{definition}

In the rest of the paper, we will refer to a behavioral strategy simply as a strategy for convenience. A strategy profile $\pi$ is the Cartesian product of all agents' strategy, i.e., $\pi=\pi^1\times\cdots\times\pi^n$. 

We denote the probability of agents using the joint action $a$ on state $s$ by $\pi(s,a)$, the probability of agent $i$ using the action $a^i$ on state $s$ by $\pi^i(s,a^i)$. The strategy profile other than agent $i$ is denoted by $\pi^{-i}$. Given $\pi$, the transition probability and the reward function only depend on the current state $s\in \bbS$. So let $r^{i,\pi}(s)$ denote $\bbE_{a\sim \pi(s)}[r^i(s,a)]$ and let $P^{\pi}(s'|s)$ denote $\bbE_{a\sim \pi(s)}[P(s'|s,a)]$. Given $\pi^{-i}$, the transition probability and the reward function only depend on the current state $s\in \bbS$ and player $i$'s action $a^i$. So let $r^{i,\pi^{-i}}(s,a^i)$ denote $\bbE_{a^{-i}\sim \pi^{-i}(s)}[r^i(s,(a^i,a^{-i}))]$ and let $P^{\pi^{-i}}(s'|s,a)$ denote $\bbE_{a^{-i}\sim \pi^{-i}(s)}[P(s'|s,(a^i,a^{-i}))]$.

For any positive integer $m$, let $\Delta_m:=\{x\in \mathcal{R}_+^m|\sum_{i=1}^m x_i=1\}$. Define $\Delta_{A^i}^k:=\times_{p=1}^k\Delta_{A^i}$. Then $\forall s\in\bbS,\pi^i(s)\in \Delta_{A^i}$, $\pi^i\in \Delta_{A^i}^S$ and $\pi\in \prod_{i=1}^n\Delta_{A^i}^S$.

\begin{definition}[Value Function]\label{definition: value function}
A value function for a strategy profile $\pi$ of an agent $i$, written $V^{\pi^i,\pi^{-i}}:\bbS\rightarrow R$ gives the expected sum of discounted rewards of the agent $i$ when the starting state is $s$: 
$$V^{\pi^i,\pi^{-i}}(s)=\bbE \left[\sum_{t=0}^{\infty} \gamma^t r^i(s_t,a)\Big|s_0=s,a\sim\pi(s_t), s_{t+1}\sim P^{\pi}(s_t) \right].$$

Alternatively, the value function can also be defined recursively via the Bellman equation.
$$V^{\pi^i,\pi^{-i}}(s)=\sum_{s'\in \bbS}\mathop{\bbE}_{a\sim\pi(s)}\left[ r^{i}(s,a)\right]+\gamma P^{\pi}(s'|s)V^{\pi^i,\pi^{-i}}(s').$$
\end{definition}

\begin{definition}[Markov Perfect Equilibrium (MPE)]\label{definition: MPE}
A behavioral strategy profile $\pi$ is called a Markov Perfect Equilibrium if $$\forall s\in \bbS, i\in[n], \forall \tilde{\pi}^i\in \Delta_{A^i}^S, V^{\pi^i,\pi^{-i}}(s)\geq V^{\tilde{\pi}^i,\pi^{-i}}(s).$$
\end{definition}

\begin{definition}[$\epsilon$-approximate MPE]
Given $\epsilon>0$, a behavioral strategy profile $\pi$ is called an $\epsilon$-approximate MPE if $$\forall s\in \bbS, i\in[n], \forall \tilde{\pi}^i\in \Delta_{A^i}^S, V^{\pi^i,\pi^{-i}}(s)\geq V^{\tilde{\pi}^i,\pi^{-i}}(s)-\epsilon.$$
\end{definition}

The Markov perfect equilibrium is a concept within SGs in which the players’ strategies depend only on the current state and not the game history. So the state encodes all relevant information for the player's strategies.

\section{The Class \textbf{PPAD} and {\sc Markov-Perfect Equilibrium} Problem}
\label{sec: PPAD and MPE problem}

The complexity class \textbf{PPAD} is introduced~\cite{papadimitriou1994complexity} to characterize the
mathematical proof structure required in a class of mathematical problems based on a parity argument for a solution to exist as in the following problem of {\sc End of the Line}.
It has included Nash equilibrium computation~\cite{daskalakis2009complexity,chen2006settling}, as well as many other problems.

The problem is defined on a class of directed graphs consisting of an exponential number of vertices (numbered from $0^n$ to $2^n-1$). Edges of this graph is defined by two polynomial-size circuits $S$ and $P$, each with $n$ input bits and $n$ output bits. There is an edge from vertex $u$ to vertex $v$ if and only if $S(u)=v$ and $P(v)=u$. Note that each vertex has at most 1 indegree and at most 1 outdegree, which means that the graph only consists of paths, cycles, and isolated vertices.
\begin{definition}[$(S,P)$-Graph~\cite{goldberg2013complexity}]
An $(S,P)$-graph with parameter $n$ is a graph on $\{0,1\}^n$ specified by circuits $S$ and $P$, as described above, subject to the constraint that vertex $0^n$ has no incoming edge but does have an outgoing edge.
\end{definition}

Based on $(S,P)$-graphs, the problem {\sc End of the Line} is to find a vertex other that $0^n$ such that
the sum of its indegree and outdegree is one but {\sc Other End of this Line} is to find the end of the particular path that starts at $0^n$~\cite{goldberg2013complexity}.
It turns out that the two problems are dramatically different in terms of their computational complexity. The former is \textbf{PPAD}-complete~\cite{papadimitriou1994complexity} but the latter is PSPACE-complete~\cite{goldberg2013complexity}.

Here we give the definition of computational problem of finding a Markov Perfect Equilibrium in Stochastic Games.

\begin{definition}[{\sc Markov-Perfect Equilibrium}]
The input instance of problem {\sc Markov-Perfect Equilibrium} is a pair $(\mathcal{SG},L)$ where $\mathcal{SG}$ is a Stochastic Game and $L$ is a binary integer. The output of problem {\sc Markov-Perfect Equilibrium} is a strategy profile $\pi\in\prod_{i=1}^n\Delta_{A^i}^S$ such that $\pi$ is a $1/L$-approximate MPE.
\end{definition}

\begin{theorem}[Main Theorem]\label{PPAD-complete}
{\sc Markov-Perfect Equilibrium} is \textbf{PPAD}-complete.
\end{theorem}

We note that when $|S|=1$ and $\gamma=0$, a Stochastic Game degenerates to an $n$-player matrix game. At this time, any Markov Perfect Equilibrium of this Stochastic Game is a Nash Equilibrium for the corresponding matrix game. So we have the following hardness result immediately:

\begin{lemma}\label{PPAD-hard}
{\sc Markov-Perfect Equilibrium} is \textbf{PPAD}-hard.
\end{lemma}

In the rest of the paper, we will mainly focus on the proof of \textbf{PPAD} membership of MPE.

\begin{lemma}\label{PPAD membership}
{\sc Markov-Perfect Equilibrium} is in \textbf{PPAD}.
\end{lemma}

\section{On the Existence of MPE}\label{Existence of MPE}

The original proof of the existence of MPE is from \cite{fink1964equilibrium}, mainly based on the Kakutani fixed point theorem. Here we give an alternative proof based on the Brouwer fixed point theorem, which also leads to our proof of \textbf{PPAD} membership of {\sc Markov-Perfect Equilibrium}.

Inspired by the continuous transformation defined by Nash to prove the existence of equilibrium point \cite{nash1951}, we define a new function $f: \prod_{i=1}^n\Delta_{A^i}^S\rightarrow \prod_{i=1}^n\Delta_{A^i}^S$ for a Stochastic Game to establish the existence of MPE. Let $V^{\pi^i,\pi^{-i}}_{\pi^i(s,a^i)=1}(s)$ denote the value function of agent $i$ if agent $i$ uses pure action $a^i$ at state $s$, uses mixed actions $\pi^i(s')$ at state $s'\neq s$, and for any other agent $j\neq i$, agent $j$ uses the strategy $\pi^j$. 

Let $\pi\in \prod_{i=1}^n\Delta_{A^i}^S$ be a strategy profile. Then for each player $i\in [n]$, each state $s\in\bbS$ and each action $a^i\in\bbA^i$, the modification of $\pi^i(s,a^i)$ is defined as follows: $$\left(f(\pi)\right)^i(s,a^i)=\frac{\pi^i(s,a^i)+\max\left(0,V^{\pi^i,\pi^{-i}}_{\pi^i(s,a^i)=1}(s)-V^{\pi^i,\pi^{-i}}(s)\right)}{1+\sum_{b^i\in\bbA^i}\max\left(0,V^{\pi^i,\pi^{-i}}_{\pi^i(s,b^i)=1}(s)-V^{\pi^i,\pi^{-i}}(s)\right)}.$$

We define the distance of two strategy profiles $\pi_1$ and $\pi_2$, denoted by $\|\pi_1-\pi_2\|_{\infty}$, as follows. $\|\pi_1-\pi_2\|_{\infty}=\max_{i\in[n], s\in\bbS,a^i\in\bbA^i}|\pi_1^i(s,a^i)-\pi_2^i(s,a^i)|$.

We first prove the function $f$ satisfies a continuity property namely \textit{$\lambda$-Lipschitz}, where $\lambda$ is defined as $\frac{9nS^2A_{\max}^2R_{\max}}{(1-\gamma)^2}$. The proof of \Cref{lemma: f Lipschitz} is challenging, because the value function $V^{\pi^i,\pi^{-i}}$ is defined recursively via Bellman equation. It could be written informally like $V^{\pi^i,\pi^{-i}}=(I-\gamma P^{\pi})^{-1}r^{i,\pi}$, which is not linear even for each fixed $\pi^{-i}$. We refer the interested reader to \Cref{proof of sec 4} for a complete proof, whose techniques might be of independent interest.

\begin{restatable}{lemma}{fLipschitz}
\label{lemma: f Lipschitz}
The function $f$ is $\lambda$-Lipschitz, i.e., for every $\pi_1,\pi_2\in \prod_{i=1}^n\Delta_{A^i}^S$ such that $\left\|\pi_1-\pi_2\right\|_{\infty}\leq \delta$, we have $$\Big\|f(\pi_1)-f(\pi_2)\Big\|_{\infty}\leq \frac{9nS^2A_{\max}^2R_{\max}}{(1-\gamma)^2} \delta.$$
\end{restatable}

Now we could establish the existence of MPE by the Brouwer fixed point theorem.
\begin{theorem}\label{theorem: MPE exists}
For any Stochastic Game $\left\langle n,\bbS,\bbA,P,R,\gamma\right\rangle$, a strategy profile $\pi$ is MPE if and only if it is a fixed point of the function $f$, i.e., $f(\pi)=\pi$. Furthermore, the function $f$ has at least one fixed point.
\end{theorem}

\begin{proof}
We first show the function $f$ has at least one fixed point. Brouwer fixed point theorem states that for any continuous function mapping a compact convex set to itself, there is a fixed point. Notice that $f$ is a function mapping a compact convex set to itself. Also, $f$ is continuous by \Cref{lemma: f Lipschitz}. So the function $f$ has at least one fixed point.

We then prove a strategy profile $\pi$ is MPE if and only if it is a fixed point.

The proof of the necessity part is immediate by the definition of MPE (\Cref{definition: MPE}). If $\pi$ is a MPE, then we have for each player $i\in [n]$, each state $s\in\bbS$ and each action $a^i\in\bbA^i$, $V^{\pi^i,\pi^{-i}}(s)\geq V^{\pi^i,\pi^{-i}}_{\pi^i(s,a^i)=1}(s)$, which means $\max\left(0,V^{\pi^i,\pi^{-i}}_{\pi^i(s,a^i)=1}(s)-V^{\pi^i,\pi^{-i}}(s)\right)=0$. Then for each player $i\in [n]$, each state $s\in\bbS$ and each action $a^i\in\bbA^i$, $\left(f(\pi)\right)^i(s,a^i)=\pi^i(s,a^i)$, which means $\pi$ is a fixed point of $f$.

For the proof of the sufficiency part, suppose that $\pi$ is a fixed point of $f$. Then we have for each player $i\in [n]$, each state $s\in\bbS$ and each action $a^i\in\bbA^i$
\begin{eqnarray*}
	&&\pi^i(s,a^i)=\frac{\pi^i(s,a^i)+\max\left(0,V^{\pi^i,\pi^{-i}}_{\pi^i(s,a^i)=1}(s)-V^{\pi^i,\pi^{-i}}(s)\right)}{1+\sum_{b^i\in\bbA^i}\max\left(0,V^{\pi^i,\pi^{-i}}_{\pi^i(s,b^i)=1}(s)-V^{\pi^i,\pi^{-i}}(s)\right)}\\
	&\Longrightarrow& \pi^i(s,a^i)\sum_{b^i\in\bbA^i}\max\left(0,V^{\pi^i,\pi^{-i}}_{\pi^i(s,b^i)=1}(s)-V^{\pi^i,\pi^{-i}}(s)\right)=\max\left(0,V^{\pi^i,\pi^{-i}}_{\pi^i(s,a^i)=1}(s)-V^{\pi^i,\pi^{-i}}(s)\right).
\end{eqnarray*}

Pick arbitrarily $$a^{i,*}\in\mathop{\arg\min}_{b^i\in\bbA^i, \pi^i(s,b^i)>0}V^{\pi^i,\pi^{-i}}_{\pi^i(s,b^i)=1}(s).$$ It is not hard to prove $\max\left(0,V^{\pi^i,\pi^{-i}}_{\pi^i(s,a^{i,*})=1}(s)-V^{\pi^i,\pi^{-i}}(s)\right)=0$, which means 
\begin{eqnarray*}
	&& \pi^i(s,a^{i,*})\sum_{b^i\in\bbA^i\setminus\{a^{i,*}\}}\max\left(0,V^{\pi^i,\pi^{-i}}_{\pi^i(s,b^i)=1}(s)-V^{\pi^i,\pi^{-i}}(s)\right)=0\\
	&\Longrightarrow& \sum_{b^i\in\bbA^i\setminus\{a^{i,*}\}}\max\left(0,V^{\pi^i,\pi^{-i}}_{\pi^i(s,b^i)=1}(s)-V^{\pi^i,\pi^{-i}}(s)\right)=0\\
	&\Longrightarrow& \forall b^i\in\bbA^i\setminus\{a^{i,*}\}, \max\left(0,V^{\pi^i,\pi^{-i}}_{\pi^i(s,b^i)=1}(s)-V^{\pi^i,\pi^{-i}}(s)\right)=0.
\end{eqnarray*}
So we have $\forall b^i\in \bbA^i, \max\left(0,V^{\pi^i,\pi^{-i}}_{\pi^i(s,b^i)=1}(s)-V^{\pi^i,\pi^{-i}}(s)\right)=0$, i.e., for any state $s\in \bbS$,
\begin{equation}\label{equation: pi^i in argmax}
\pi^i\in\mathop{\arg\max}_{\substack{\pi^{i,*}\in \Delta_A^S\\\forall s'\neq s,\pi^{i,*}(s')=\pi^i(s)}}V^{\pi^{i,*},\pi^{-i}}(s).
\end{equation}

Note that if we fix the strategy profile other agent $i$, then for agent $i$, it is essentially a Markov decision process. By \Cref{equation: pi^i in argmax}, we know that $\pi^i$ is an optimal policy of agent $i$, which means $$\forall s\in \bbS, i\in[n], \forall \tilde{\pi}^i\in \Delta_A^S, V^{\pi^i,\pi^{-i}}(s)\geq V^{\tilde{\pi}^i,\pi^{-i}}(s),$$ i.e., $\pi$ is a MPE of the Stochastic Game.
\end{proof}

\section{\textbf{PPAD} Membership of {\sc Markov-Perfect Equilibrium}}\label{section: membership}

In this section, we will prove the \textbf{PPAD} membership of {\sc Markov-Perfect Equilibrium}, by reducing it to {\sc End of the Line}. We highlight our approximation guarantee proof (\Cref{appoximation gua}), which includes several innovative understanding of Markov Decision Processes and Stochastic Games. The construction of the graph of {\sc End of the Line} is relatively standard and is from the simplicial approximation algorithm of Laan and Talman \cite{LaanT82computation}, which will be provided into \Cref{construct graph}.

\subsection{The Approximation Guarantee}\label{appoximation gua}
In \Cref{Existence of MPE}, \Cref{theorem: MPE exists} states that $f$ has a fixed point $\pi$ if and only if $\pi$ is an MPE for the Stochastic Game. Now we will prove $f$ has some good approximation properties beyond that: if we find an $\epsilon$-approximate fixed point $\pi$ of $f$, then it is also a $\texttt{poly}(|\mathcal{SG}|)\sqrt{\epsilon}$-approximate MPE for the Stochastic Game (combining \Cref{lemma: single state} and \Cref{lemma: all states}). 

Moreover, we also get \Cref{coro: approxi MDP}, which leads to better understanding for Markov Decision Process and might be of independent interest. The statement of \Cref{coro: approxi MDP} is as follows. Let $\epsilon>0$ and $\pi$ be a (not necessarily deterministic) policy. If for every starting state $s_0\in\bbS$, the agent only changes the action of $s_0$ could gain at most $\epsilon$ more value, then the agent could gain at most $\epsilon/(1-\gamma)$ more value even if the agent changes its policy to the optimal policy, i.e., $\pi$ is a good approximation of MDP.

The formal statements of lemmas and proofs are as follows. Proof of \Cref{lemma: single state} is in \Cref{proof of single state}.

\begin{restatable}{lemma}{apprsinglestate}\label{lemma: single state}
Let $\epsilon>0$ and $\pi$ be a strategy profile. If $\|f(\pi)-\pi\|_{\infty}\leq \epsilon$, then for each player $i\in [n]$, each state $s\in \bbS$ and each action $a^i\in \bbA^i$, we have $$\max\left(0,V^{\pi^i,\pi^{-i}}_{\pi^i(s,a^i)=1}(s)-V^{\pi^i,\pi^{-i}}(s)\right)\leq A_{\max}\left(\frac{\sqrt{\epsilon'}}{1-\gamma}+R_{\max}\sqrt{\epsilon'}+\epsilon'\right),$$ where $\epsilon'=\epsilon\left(1+\dfrac{A_{\max}R_{\max}}{1-\gamma}\right).$
\end{restatable}

\begin{lemma}\label{lemma: all states}
Let $\epsilon>0$ and $\pi$ be a strategy profile. If for each player $i\in [n]$, each state $s\in \bbS$ and each action $a^i\in \bbA^i$, $\max\left(0,V^{\pi^i,\pi^{-i}}_{\pi^i(s,a^i)=1}(s)-V^{\pi^i,\pi^{-i}}(s)\right)\leq \epsilon$, then $\pi$ is an $\epsilon/(1-\gamma)$-approximate MPE.
\end{lemma}
\begin{proof}
Pick any player $i\in [n]$, it is sufficient for us to prove $\forall s\in \bbS, \forall \tilde{\pi}^i\in \Delta_A^S, V^{\pi^i,\pi^{-i}}(s)\geq V^{\tilde{\pi}^i,\pi^{-i}}(s)-\epsilon.$ Suppose that $\max_{a^i\in\bbA^i}\left(0,V^{\pi^i,\pi^{-i}}_{\pi^i(s,a^i)=1}(s)-V^{\pi^i,\pi^{-i}}(s)\right)=\epsilon(s).$ Consider the following linear program:
\begin{equation}\label{equation: optimal lp}
\begin{aligned}
	\min \quad& \sum_{s\in\bbS} V(s)& \\
	\text{s.t.,} \quad &  V(s)\geq r^{i,\pi^{-i}}(s,a^i)+\gamma\sum_{s'\in\bbS}P^{\pi^{-i}}(s'|s,a^i)V(s')&\forall s\in \bbS, a^i\in \bbA^i.
\end{aligned}	
\end{equation}
Let $V^*$ be the solution of the linear program (\ref{equation: optimal lp}). It satisfies $$V^*(s)=\max_{a^i\in\bbA^i}\left(r^{i,\pi^{-i}}(s,a^i)+\gamma\sum_{s'\in\bbS}P^{\pi^{-i}}(s'|s,a^i)V^*(s')\right),$$ which is also the value function of player $i$ when she uses the optimal policy given others' strategy profile $\pi^{-i}$. (Note that when we are given $\pi^{-i}$, it is essentially a Markov Decision Process for player $i$. So we are using linear programming to solve this MDP.)

Now look at the other linear program:
\begin{equation}\label{equ: approxmate lp}
\begin{aligned}
	\min \quad& \sum_{s\in\bbS} V(s)& \\
	\text{s.t.,} \quad &  V(s)\geq r^{i,\pi^{-i}}(s,a^i)-\epsilon(s)+\gamma\sum_{s'\in\bbS}P^{\pi^{-i}}(s'|s,a^i)V(s')&\forall s\in \bbS, a^i\in \bbA^i.
\end{aligned}	
\end{equation}

Let $V'$ be the solution of the linear program (\ref{equ: approxmate lp}). It satisfies $$V'(s)=\max_{a^i\in\bbA^i}\left(r^{i,\pi^{-i}}(s,a^i)+\gamma\sum_{s'\in\bbS}P^{\pi^{-i}}(s'|s,a^i)V'(s')\right)-\epsilon(s),$$ which is also the value function for the strategy profile $\pi$ for the player $i$.

Now it is sufficient for us to bound $V^*(s)-V'(s), \forall s\in\bbS$. Let $\epsilon_{\max}=\max_{s\in\bbS}\epsilon(s)$. Construct a new value vector for the player $i$: $\tilde{V}(s)=V'(s)+\epsilon_{\max}/(1-\gamma)$. Then we have

\begin{eqnarray*}
	&&V(s)\geq r^{i,\pi^{-i}}(s,a^i)-\epsilon(s)+\gamma\sum_{s'\in\bbS}P^{\pi^{-i}}(s'|s,a^i)V(s')\\
	&\Longleftrightarrow& V'(s)+\frac{\epsilon_{\max}}{1-\gamma}\geq r^{i,\pi^{-i}}(s,a^i)-\epsilon(s)+\frac{\epsilon_{\max}}{1-\gamma}+\gamma\sum_{s'\in\bbS}P^{\pi^{-i}}(s'|s,a^i)V(s')\\
	&\Longleftrightarrow& V'(s)+\frac{\epsilon_{\max}}{1-\gamma}\geq r^{i,\pi^{-i}}(s,a^i)-\epsilon(s)+\epsilon_{\max}+\gamma\sum_{s'\in\bbS}P^{\pi^{-i}}(s'|s,a^i)\left(V(s')+\frac{\epsilon_{\max}}{1-\gamma}\right)\\
	&\Longleftrightarrow& \tilde{V}(s)\geq r^{i,\pi^{-i}}(s,a^i)-\epsilon(s)+\epsilon_{\max}+\gamma\sum_{s'\in\bbS}P^{\pi^{-i}}(s'|s,a^i)\tilde{V}(s)\\
	&\Longrightarrow& \tilde{V}(s)\geq r^{i,\pi^{-i}}(s,a^i)+\gamma\sum_{s'\in\bbS}P^{\pi^{-i}}(s'|s,a^i)\tilde{V}(s).
\end{eqnarray*}

So $\tilde{V}$ is a feasible solution of linear program (\ref{equation: optimal lp}), which means $V^*(s)\leq \tilde{V}(s)$ for any $s\in\bbS$. Then we have $$V^*(s)-V'(s)\leq \tilde{V}(s)-V'(s)=\epsilon_{\max}/(1-\gamma),$$ i.e., the difference between the optimal value $V^*(s)$ and $V^{\pi^i,\pi^{-i}}$ is upper bounded by $\epsilon/(1-\gamma)$. The argument above applies to any player. So by the definition of $\epsilon$-approximate MPE, we know that $\pi$ is an $\epsilon/(1-\gamma)$-approximate MPE.
\end{proof}

\begin{corollary}\label{coro: approxi MDP}
Let $\epsilon>0$ and $\pi$ be a (not necessarily deterministic) policy of the agent. If for each state $s\in \bbS$ and each action $a\in \bbA$ (where $\bbA$ is the action space of the agent), $\max\left(0,V^{\pi}_{\pi(s,a)=1}(s)-V^{\pi}(s)\right)\leq \epsilon$, then $\pi$ is an $\epsilon/(1-\gamma)$ approximation of MDP.
\end{corollary}

\subsection{Constructing the {\sc End of the Line} Graph}\label{construct graph}

In this section, we give an outline of our reduction from {\sc Markov-Perfect Equilibrium} to {\sc End of the Line}, with the help of the simplicial approximation algorithm of Laan and Talman~\cite{LaanT82computation}. We will focus on the correctness of reduction, leaving details about how to construct the vertices to the appendix.

Recall that the input instance of {\sc Markov-Perfect Equilibrium} is a pair $(\mathcal{SG},L)$. Let $d$ be an integer, which will be defined later to make sure we can find an $1/L$-approximate MPE.

For each $i\in[n]$, define $\Delta_{A^i}(d)$ is the set of points of $\Delta_{A^i}$ induced by the regular grid of size $d$, i.e., $$\Delta_{A^i}(d)=\left\{x\in\Delta_{A^i}|x_j=y_j/d,y_j\in\mathbb{Z}^+,\sum_{j=1}^{A^i}y_j=d \right\}.$$ Similarly, define $\Delta_{A^i}^k(d):=\times_{p=1}^k\Delta_{A^i}(d)$.

\textbf{The Vertices of {\sc End of the Line} Graph.} The set of vertices $\Sigma$ is a set of simplices defined on $\prod_{i=1}^n \Delta_{A^i}^S(d)$, which could be encoded with string $\{0,1\}^N$, where $N$ is polynomial in $|\mathcal{SG}|$ and $\log d$. The formal definition of $\Sigma$ is in \Cref{vertices of end of the line}.

\textbf{Labelling the Grid Points.} We will give each point in $\prod_{i=1}^n \Delta_{A^i}^S(d)$ a label, which will be an element of the set $\mathcal{L}:=\bigcup_{i\in[n],s\in\bbS,a^i\in\bbA^i}(i,s,a^i)$.

Without loss of generality, we assign a number to the state set $\bbS$ and action set $\bbA^i$ for each $i\in[n]$ arbitrarily for the purpose of labelling. Suppose that $\bbS=\{s_1,\cdots,s_S\}$ and $\bbA^i=\{a^i_1,\cdots,a^i_{A^i}\}$. 

For each strategy profile $\pi\in \prod_{i=1}^n \Delta_{A^i}^S(d)$, $\pi$ receives the label $(i,s_j,a^i_k)$ if and only if $(i,s_j,a^i_k)$ is the lexicographically least index such that $\pi^i(s_j,a^i_k)>0$ and $$(f(\pi))^i(s_j,a^i_k)-\pi^i(s_j,a^i_k)\leq (f(\pi))^{i'}(s_{j'},a^{i'}_{k'})-\pi^{i'}(s_{j'},a^{i'}_{k'})$$ for all $i'\in[n],s_{j'}\in\bbS$ and $a^{i'}_{k'}\in\bbA^{i'}$.

Note that each strategy profile $\pi\in \prod_{i=1}^n \Delta_{A^i}^S(d)$ has exactly one label, which could be denoted by $l(\pi)$. Since the function $f$ could be computed in time polynomial in $N$ and $|\mathcal{SG}|$, the label could also computed in time polynomial in $|\mathcal{SG}|$ and $\log d$. Also the labelling rule is proper in the sense that $l(\pi)\neq (i,s_j,a^i_k)$ if $\pi^i(s_j,a^i_k)=0$.

A simplex $\sigma\in\Sigma$ will be called complete labelled if all its vertices\footnote{Please distinguish the vertices of a simplex and vertices of the {\sc End of the Line} graph.} have a different label. A completely labelled simplex $\sigma$ is called $(i,s_j)$-stopping if for each $a^i_k\in\bbA^i$, there exists $\pi\in\sigma$ such that $l(\pi)=(i,s_j,a^i_k)$. Further, a completely labelled simplex $\sigma$ is called stopping if there exist $i\in[n]$ and $s_j\in\bbS$ such that $\sigma$ is $(i,s_j)$-stopping.

The following lemma asserts that if we can find a stopping simplex, then we can find an $\texttt{poly}(|\mathcal{SG}|)/d$-approximate fixed point. The proof is in \Cref{proof of stopping appro}.

\begin{restatable}{lemma}{approstopping}{(\rm \cite{LaanT82computation})}\label{lemma: stopping appro}
Suppose that a simplex $\sigma\in\Sigma$ is $(i,s)$-stopping for $i\in[n]$ and $s\in\bbS$. Then for any strategy profile $\pi\in\sigma$, we have $$\|f(\pi)-\pi\|_{\infty}\leq A_{\max}^2(\lambda+1)\frac{1}{d}.$$
\end{restatable}

\textbf{The Choice of $d$.} Let $$d=\dfrac{32A_{\max}^5R_{\max}^3(\lambda+1)}{(1-\gamma)^5}L^2.$$ It is easy to see $d$ is $\texttt{poly}(|\mathcal{SG}|,L)$. The correctness of our choice is in \Cref{correct choice of d}.

\textbf{The Edges of {\sc End of the Line} Graph.} In the algorithm of Laan and Talman~\cite{LaanT82computation}, they develop a partial one-to-one function $g: \Sigma'\rightarrow\Sigma'$ for $\Sigma'\subseteq\Sigma$ as well as a starting simplex $\sigma_0\in\Sigma$, which have the following properties:
\begin{itemize}
    \item $\sigma_0\in\Sigma'$ and there is no $\sigma'\in\Sigma'$ such that $g(\sigma')=\sigma_0$;
    \item For any $\sigma\in\Sigma'$, if $\sigma$ has no image, then $\sigma$ is a stopping simplex. For any $\sigma\in\Sigma'\setminus\{\sigma_0\}$, if $\sigma$ has no pre-image, then $\sigma$ is a stopping simplex. 
    \item the function $g$ and $g^{-1}$ could be computed in time polynomial in $|\mathcal{SG}|$ and $\log d$.
\end{itemize}

For the purpose of constructing the {\sc End of the Line} graph, we complete the function $g$ by letting $g(\sigma)=\sigma$ for any $\sigma\in\Sigma\setminus\Sigma'$. It is easy to verify our operation does not impact the properties of function $g$. So for any input instance $(|\mathcal{SG}|,L)$, we can reduce it to an instance of {\sc End of the Line}, where the two circuits $S$ and $P$ correspond to $g$ and $g^{-1}$. If we can find a solution of the {\sc End of the Line}, by \Cref{lemma: stopping appro} we know that there is an $A_{\max}^2(\lambda+1)\frac{1}{d}$-approximate fixed point in the solution simplex, thus an $1/L$-approximate MPE by \Cref{lemma: single state}, \Cref{lemma: all states}, and our choice of $d$.

\section{Conclusion}

Solving a Markov Perfect Equilibrium (MPE) in general-sum stochastic games (SG) has long expected to be at least $\textbf{PPAD}$-hard. 
In this paper, 
we prove that computing an MPE in a finite-state infinite horizon discounted SGs is $\textbf{PPAD}$-complete. 
Our proof is novel in the sense that we adopt a function with a polynomial-bound description in the strategy space that effectively helps convert the MPE computation problem to a fixed-point problem, which, otherwise, would take a representation that requires an exponential number of pure strategies with respect to the number of states and the number of agents.   
Our completeness result indicates that computing MPE in SGs is highly unlikely to be $\textbf{NP}$-hard.  
We hope our results can encourage MARL researchers to study solving MPE  in general-sum SGs,  leading to more prosperous algorithmic developments as those currently on zero-sum SGs.

\clearpage


\bibliographystyle{plain}
\bibliography{main}

\clearpage

\appendix

\addcontentsline{toc}{section}{Appendix} 
\part{{\Large{Appendix for \emph{"On the Complexity of Computing Markov Perfect Equilibrium in General-Sum Stochastic Games"}}}} 
\parttoc

\section{Detailed Proofs from \Cref{Existence of MPE}}
\label{proof of sec 4}

\subsection{Proof of \Cref{lemma: f Lipschitz}}

\fLipschitz*
\begin{proof}[Proof of \Cref{lemma: f Lipschitz}]
We first give an upper bound of $\left|r^{i,\pi_1}(s)-r^{i,\pi_2}(s)\right|$ for any $s\in \bbS$ and $i\in[n]$.
\begin{eqnarray*}
	&&\left|r^{i,\pi_1}(s)-r^{i,\pi_2}(s)\right|\\
	&=& \left|\sum_{a\in\bbA}r^i(s,a)\pi_1(s,a)-\sum_{a\in\bbA}r^i(s,a)\pi_2(s,a)\right|\\
	&=& \left|\sum_{a\in\bbA}r^i(s,a)\prod_{i\in[n]}\pi_1^i(s,a^i)-\sum_{a\in\bbA}r^i(s,a)\prod_{i\in[n]}\pi_2^i(s,a^i)\right|\\
	&\leq& \sum_{a\in\bbA}\left|r^i(s,a)\right|\left|\prod_{i\in[n]}\pi_1^i(s,a^i)-\prod_{i\in[n]}\pi_2^i(s,a^i)\right|\\
	&\leq& R_{\operatorname{max}}\sum_{a\in\bbA}\left|\prod_{i\in[n]}\pi_1^i(s,a^i)-\prod_{i\in[n]}\pi_2^i(s,a^i)\right|\\
	&\leq& nA_{\max}R_{\operatorname{max}}\delta,
\end{eqnarray*}
where the last inequality follows
\begin{eqnarray*}
	&&\sum_{a\in\bbA}\left|\prod_{i\in[n]}\pi_1^i(s,a^i)-\prod_{i\in[n]}\pi_2^i(s,a^i)\right|\\
	&=& \sum_{a\in\bbA}\left|\sum_{k=1}^n\prod_{i=1}^{k-1}\pi_1^i(s,a^i)\left(\pi_1^k(s,a^k)-\pi_2^k(s,a^k)\right)\prod_{i=k+1}^{n}\pi_2^i(s,a^i)\right|\\
	&\leq& \delta\sum_{a\in\bbA}\left|\sum_{k=1}^n\prod_{i=1}^{k-1}\pi_1^i(s,a^i)\prod_{i=k+1}^{n}\pi_2^i(s,a^i)\right|\\
	&\leq& nA_{\max}\delta.
\end{eqnarray*}

Let $V^{\pi^i,\pi^{-i}}$ denote the column vector $(V^{\pi^i,\pi^{-i}}(s))_{s\in\bbS}$, $r^{i,\pi}$ denote the column vector $(r^{i,\pi}(s))_{s\in\bbS}$, and $P^{\pi}$ denote the matrix $P^{\pi}(s,s')_{s,s'\in \bbS}$ respectively. By the Bellman equation (\Cref{definition: value function}), we have $$V^{\pi^i,\pi^{-i}}=r^{i,\pi}+\gamma P^{\pi}V^{\pi^i,\pi^{-i}},$$ which means $$V^{\pi^i,\pi^{-i}}=(I-\gamma P^{\pi})^{-1}r^{i,\pi}.$$ We will prove that $\left|(I-\gamma P^{\pi_1})^{-1}(s'|s)-(I-\gamma P^{\pi_2})^{-1}(s'|s)\right|\leq \frac{nSA_{\max}\delta}{(1-\gamma)^2}$ for any $s,s'\in\bbS$ in the following lemma (\Cref{lemma: I-gammaP inverse}).

Now we could give an upper bound of $\left|V^{\pi^i_1,\pi^{-i}_1}(s)-V^{\pi^i_2,\pi^{-i}_2}(s)\right|$ for any $s\in \bbS$.
\begin{eqnarray*}
	&&\left|V^{\pi^i_1,\pi^{-i}_1}(s)-V^{\pi^i_2,\pi^{-i}_2}(s)\right|\\
	&=& \left|\sum_{s'\in\bbS}r^{i,\pi_1}(s')(I-\gamma P^{\pi_1})^{-1}(s'|s)-\sum_{s'\in\bbS}r^{i,\pi_2}(s')(I-\gamma P^{\pi_2})^{-1}(s'|s)\right|\\
	&=& \left|\sum_{s'\in\bbS}r^{i,\pi_1}(s')\left((I-\gamma P^{\pi_1})^{-1}(s'|s)-(I-\gamma P^{\pi_2})^{-1}(s'|s)\right)+(I-\gamma P^{\pi_2})^{-1}(s'|s)\left(r^{i,\pi_1}(s')-r^{i,\pi_2}(s')\right)\right|\\
	&\leq& \sum_{s'\in\bbS}\left(R_{\max}\frac{nSA_{\max}\delta}{(1-\gamma)^2}+\frac{1}{1-\gamma}nA_{\max}R_{\max}\delta\right)\\
	&=& \frac{nSA_{\max}R_{\max}}{1-\gamma}\left(1+\frac{S}{1-\gamma}\right)\delta\\
	&\leq& \frac{2nS^2A_{\max}R_{\max}}{(1-\gamma)^2}\delta
\end{eqnarray*}
where the forth line follows from $|(I-\gamma P^{\pi_2})^{-1}(s'|s)|\leq \frac{1}{1-\gamma}$, which will also be proved in \Cref{lemma: I-gammaP inverse}.

Let $D^{\pi^i,\pi^{-i}}_{\pi^i(s,a^i)=1}(s)=\max\left(0,V^{\pi^i,\pi^{-i}}_{\pi^i(s,a^i)=1}(s)-V^{\pi^i,\pi^{-i}}(s)\right)$. Then we have the following upper bounds directly $$\left|D^{\pi^i_1,\pi^{-i}_1}_{\pi^i_1(s,a^i)=1}(s)-D^{\pi^i_2,\pi^{-i}_2}_{\pi^i_2(s,a^i)=1}(s)\right|\leq \frac{4nS^2A_{\max}R_{\max}}{(1-\gamma)^2}\delta,$$ $$\left|\sum_{b^i\in\bbA^i}\left(D^{\pi^i_1,\pi^{-i}_1}_{\pi^i_1(s,b^i)=1}(s)-D^{\pi^i_2,\pi^{-i}_2}_{\pi^i_2(s,b^i)=1}(s)\right)\right|\leq \frac{4nS^2A_{\max}^2R_{\max}}{(1-\gamma)^2}\delta.$$

Finally, we could complete our proof of this lemma. For any player $i\in[n]$, any state $s\in\bbS$ and any action $a^i\in\bbA^i$, we have 
\begin{eqnarray*}
	&&\left|(f(\pi_1))^i(s,a^i)-(f(\pi_2))^i(s,a^i)\right|\\
	&\leq& \left|\pi^i_1(s,a^i)-\pi^i_2(s,a^i)\right|+\left|D^{\pi^i_1,\pi^{-i}_1}_{\pi^i_1(s,a^i)=1}(s)-D^{\pi^i_2,\pi^{-i}_2}_{\pi^i_2(s,a^i)=1}(s)\right|+\left|\sum_{b^i\in\bbA^i}D^{\pi^i_1,\pi^{-i}_1}_{\pi^i_1(s,b^i)=1}(s)-D^{\pi^i_2,\pi^{-i}_2}_{\pi^i_2(s,b^i)=1}(s)\right|\\
	&\leq& \delta+\frac{4nS^2A_{\max}R_{\max}}{(1-\gamma)^2}\delta+\frac{4nS^2A_{\max}^2R_{\max}}{(1-\gamma)^2}\delta\\
	&\leq& \frac{9nS^2A_{\max}^2R_{\max}}{(1-\gamma)^2}\delta.
\end{eqnarray*}
\end{proof}

\subsection{Proof of \Cref{lemma: I-gammaP inverse}}

\begin{lemma}\label{lemma: I-gammaP inverse}
For every $\pi_1,\pi_2\in \prod_{i=1}^n\Delta_{A^i}^S$ such that $\left\|\pi_1-\pi_2\right\|_{\infty}\leq \delta$, we have $$\left|(I-\gamma P^{\pi_1})^{-1}(s'|s)-(I-\gamma P^{\pi_2})^{-1}(s'|s)\right|\leq \frac{nSA_{\max}\delta}{(1-\gamma)^2}$$ for any $s,s'\in\bbS$.
\end{lemma}
\begin{proof}
We first give an upper bound of $\left|P^{\pi_1}(s'|s)-P^{\pi_2}(s'|s)\right|$ for any $s,s'\in \bbS$. 
\begin{eqnarray*}
	&&\left|P^{\pi_1}(s'|s)-P^{\pi_2}(s'|s)\right|\\
	&=& \left|\sum_{a\in\bbA}P(s'|s,a)\prod_{i\in[n]}\pi_1^i(s,a^i)-\sum_{a\in\bbA}P(s'|s,a)\prod_{i\in[n]}\pi_2^i(s,a^i)\right|\\
	&\leq& \sum_{a\in\bbA}P(s'|s,a)\left|\prod_{i\in[n]}\pi_1^i(s,a^i)-\prod_{i\in[n]}\pi_2^i(s,a^i)\right|\\
	&\leq& nA_{\max}\delta
\end{eqnarray*}

Now we view $P^{\pi}$ as an $S\times S$ matrix. For any two $S\times S$ matrices $M^1, M^2$, we use $\|M^1-M^2\|_{\max}$ to denote $\max_{i,j}|M^1(i,j)-M^2(i,j)|$, i.e., the max norm. Then we have $\|P^{\pi_1}-P^{\pi_2}\|_{\max}\leq nA_{\max}\delta$.

Let $Q^1=(I-\gamma P^{\pi_1})^{-1}$ and $Q^2=(I-\gamma P^{\pi_2})^{-1}$. (Note that the inverse of $(I-\gamma P^{\pi})$ must exist because $\gamma<1$.) 

By definition, we have $Q^1=I+\gamma P^{\pi_1}Q^1$ and $Q^2=I+\gamma P^{\pi_2}Q^2$. Then
\begin{eqnarray*}
	&&\left\|Q^1-Q^2\right\|_{\max}\\
	&=& \gamma\left\|P^{\pi_1}Q^1-P^{\pi_2}Q^2\right\|_{\max}\\
	&=& \gamma \max_{i,j}\left|\sum_{k}P^{\pi_1}(i,k)Q^1(k,j)-\sum_{k}P^{\pi_2}(i,k)Q^2(k,j)\right|\\
	&\leq& \gamma \max_{i,j}\sum_{k}\left|P^{\pi_1}(i,k)Q^1(k,j)-P^{\pi_2}(i,k)Q^2(k,j)\right|\\ 
	&\leq& \gamma \max_{i,j}\left(\sum_{k}P^{\pi_1}(i,k)\left|Q^1(k,j)-Q^2(k,j)\right|+\sum_{k}|Q^2(k,j)|\left|P^{\pi_1}(i,k)-P^{\pi_2}(i,k)\right|\right)\\ 
	&\leq& \gamma \max_{i,j}\left(\max_{k}\left|Q^1(k,j)-Q^2(k,j)\right|+\sum_{k}\frac{nA_{\max}\delta}{1-\gamma}\right)\\ 
	&=& \gamma \left(\left\|Q^1-Q^2\right\|_{\max}+\frac{nSA_{\max}\delta}{1-\gamma}\right)
\end{eqnarray*}
where the sixth line follows the following facts:
\begin{enumerate}
    \item $\sum_k P^{\pi_1}(i,k)=1$.
    \item $\left|Q^1(k,j)-Q^2(k,j)\right|\leq \max_k \left|Q^1(k,j)-Q^2(k,j)\right|$.
    \item $\left|P^{\pi_1}(i,k)-P^{\pi_2}(i,k)\right|\leq nA_{\max}\delta$.
    \item $|Q^2(k,j)|\leq \|Q^2\|_1\leq \frac{1}{1-\gamma \|P^{\pi_2}\|_1}\leq \frac{1}{1-\gamma}$.
\end{enumerate}

Note that $Q^2=I+\gamma P^{\pi_2}Q^2$. Since 1-norm is submultiplicative, so we have $$\|Q^2\|_1\leq 1+\gamma \|P^{\pi_2}Q^2\|_1\leq 1+\gamma \|P^{\pi_2}\|_1\|Q^2\|_1\leq 1+\gamma\|Q^2\|_1,$$ which leads to the fourth fact.

So we have $$|Q^1-Q^2|_{\max}\leq \frac{nSA_{\max}\delta}{(1-\gamma)^2}.$$

\end{proof}

\section{Detailed Proofs from \Cref{section: membership}}
\label{proof of sec 5}
\subsection{Proof of \Cref{lemma: single state}}
\label{proof of single state}

\apprsinglestate*
\begin{proof}[Proof of \Cref{lemma: single state}]
Pick any player $i\in[n]$ and state $s\in\bbS$ in this proof. Suppose that the action space of player $i$ is $\bbA^i=\{a^i_1,\cdots,a^i_{A^i}\}$. For the simplicity of notations, for any $a^i_j\in\bbA^i$, let $$V_{a^i_j}(s):=V^{\pi^i,\pi^{-i}}_{\pi^i(s,a^i_j)=1}(s),$$ and $$D_{a^i_j}(s):=\max\left(0,V_{a^i_j}(s)-V^{\pi^i,\pi^{-i}}(s)\right).$$ Without loss of generality, assume that $$V_{a^i_1}(s)\geq V_{a^i_2}(s)\geq \cdots \geq V_{a^i_k}(s)\geq V^{\pi^i,\pi^{-i}}(s)\geq V_{a^i_{k+1}}(s)\geq \cdots \geq V_{a^i_{A^i}}(s).$$

We first give an upper bound of $D_{a^i_j}(s)$. $$D_{a^i_j}(s)=\max\left(0,V_{a^i_j}(s)-V^{\pi^i,\pi^{-i}}(s)\right)\leq V_{a^i_j}(s)\leq R_{\max}/(1-\gamma).$$

For any action $a^i_j\in\bbA^i$, by the condition $\|f(\pi)-\pi\|_{\infty}\leq \epsilon$, we know that 

\begin{eqnarray*}
	&&\pi^i(s,a^i_j)-\frac{\pi^i(s,a^i_j)+D_{a^i_j}(s)}{1+\sum_{b^i\in\bbA^i}D_{b^i(s)}}\leq \epsilon\\
	&\Longrightarrow& \pi^i(s,a^i_j)\sum_{b^i\in\bbA^i}D_{b^i}(s)\leq D_{a^i_j}(s)+ \epsilon\left(1+\sum_{b^i\in\bbA^i}D_{b^i}(s)\right)\\
	&\Longrightarrow& \pi^i(s,a^i_j)\sum_{b^i\in\bbA^i}D_{b^i}(s)\leq D_{a^i_j}(s)+ \epsilon\left(1+\frac{A_{\max}R_{\max}}{1-\gamma}\right).
\end{eqnarray*}

Setting $\epsilon'=\epsilon\left(1+\frac{A_{\max}R_{\max}}{1-\gamma}\right),$ we have the following crucial inequality: 
\begin{eqnarray}\label{eq: crucial ineq}
\pi^i(s,a^i_j)\sum_{b^i\in\bbA^i}D_{b^i}(s)\leq D_{a^i_j}(s)+ \epsilon'.
\end{eqnarray}

Let $t:=\sum_{j=k+1}^{A^i}\pi^i(s,a^i_j).$

\textbf{Case 1:} $t\geq \sqrt{\epsilon'}/R_{\max}$.

Note that for $k+1\leq j\leq A^i,D_{a^i_j}(s)=0$. By the inequality (\ref{eq: crucial ineq}), we have 
\begin{eqnarray*}
&&\sum_{j=k+1}^{A^i}\left(\pi^i(s,a^i_j)\sum_{b^i\in\bbA^i}D_{b^i}(s)\right)\leq \sum_{j=k+1}^{A^i}\left(D_{a^i_j}(s)+ \epsilon'\right)\\
&\Longrightarrow& t\sum_{b^i\in\bbA^i}D_{b^i}(s)\leq A_{\max}\epsilon'\\
&\Longrightarrow& D_{a^i_1}(s)\leq\sum_{b^i\in\bbA^i}D_{b^i}(s)\leq A_{\max}R_{\max}\sqrt{\epsilon'}.
\end{eqnarray*}

\textbf{Case 2:} $t\leq \sqrt{\epsilon'}/R_{\max}$.

By the inequality (\ref{eq: crucial ineq}), we have 
\begin{eqnarray*}
&&\pi^i(s,a^i_j)\sum_{b^i\in\bbA^i}D_{b^i}(s)\leq D_{a^i_j}(s)+ \epsilon'\\
&\Longrightarrow& \pi^i(s,a^i_j)^2\sum_{b^i\in\bbA^i}D_{b^i}(s)\leq \pi^i(s,a^i_j)\left(D_{a^i_j}(s)+ \epsilon'\right)\\
&\Longrightarrow& \sum_{j=1}^{A^i}\left(\pi^i(s,a^i_j)^2\sum_{b^i\in\bbA^i}D_{b^i}(s)\right)\leq \sum_{j=1}^{A^i} \left(\pi^i(s,a^i_j)\left(D_{a^i_j}(s)+ \epsilon'\right)\right)\\
&\Longrightarrow& \sum_{j=1}^{A^i}\pi^i(s,a^i_j)^2\sum_{b^i\in\bbA^i}D_{b^i}(s)\leq \sum_{j=1}^{k} \left(\pi^i(s,a^i_j)D_{a^i_j}(s)\right)+ \epsilon'\\
&\Longrightarrow& \sum_{j=1}^{A^i}\pi^i(s,a^i_j)^2\sum_{b^i\in\bbA^i}D_{b^i}(s)\leq \frac{R_{\max}}{1-\gamma}\sum_{j=1}^{k} \pi^i(s,a^i_j)+ \epsilon'\\
&\Longrightarrow& \sum_{j=1}^{A^i}\pi^i(s,a^i_j)^2\sum_{b^i\in\bbA^i}D_{b^i}(s)\leq \frac{R_{\max}}{1-\gamma}\frac{\sqrt{\epsilon'}}{R_{\max}}+ \epsilon'\\
&\Longrightarrow& \frac{1}{A^i}\sum_{b^i\in\bbA^i}D_{b^i}(s)\leq \frac{\sqrt{\epsilon'}}{1-\gamma}+ \epsilon'\\
&\Longrightarrow& D_{a^i_1}(s)\leq\sum_{b^i\in\bbA^i}D_{b^i}(s)\leq A_{\max}\left(\frac{\sqrt{\epsilon'}}{1-\gamma}+ \epsilon'\right).
\end{eqnarray*}

Note that the argument above could be applied to any player and any state, so for each player $i\in [n]$, each state $s\in \bbS$ and each action $a^i\in \bbA^i$, we have $$\max\left(0,V^{\pi^i,\pi^{-i}}_{\pi^i(s,a^i)=1}(s)-V^{\pi^i,\pi^{-i}}(s)\right)\leq A_{\max}\left(\frac{\sqrt{\epsilon'}}{1-\gamma}+R_{\max}\sqrt{\epsilon'}+\epsilon'\right).$$
\end{proof}

\subsection{Definition of the Set of Vertices $\Sigma$}
\label{vertices of end of the line}
Note that the strategy profile space is $\prod_{i=1}^n \Delta_{A^i}^S$, which is a production of unit simplices. We will adopt the techniques in~\cite{LaanT82computation} to triangulate the strategy profile space, where the set of vertices of {\sc End of the Line} graph will correspond to a set of simplices after triangulation.

We use $Q_{\bbA^i}$ to denote the $A^i\times A^i$ matrix 
$$
Q_{\bbA^i}=\left[
\begin{matrix}
-1 & 0 & . & . & . & 1 \\
1 & -1 & 0 & . & . & 0 \\
. &  . & . &   &   & . \\
. &    & . & . &   & . \\
. &    &   & . & . & . \\
0 &  . & . & 0 & 1 & -1
\end{matrix}
\right].
$$

For each player $i\in[n]$, we use $Q_i$ to denote the $A^i S\times A^i S$ matrix, which is a block diagonal matrix
$$
\left. Q_i=\left[
\begin{matrix}
Q_{\bbA^i} & 0          & \ldots & 0\\
0          & Q_{\bbA^i} & \ldots & 0\\
\vdots     & \vdots     & \ddots & \vdots \\
0          & 0          & \ldots & Q_{\bbA^i}
\end{matrix}
\right]
\right\}S.
$$

Finally, we use $Q$ to denote the block diagonal matrix
$$
Q=\left[
\begin{matrix}
Q_1 & 0          & \ldots & 0\\
0          & Q_2 & \ldots & 0\\
\vdots     & \vdots & \ddots & \vdots \\
0          & 0          & \ldots & Q_n
\end{matrix}
\right].
$$

For any agent $i\in[n]$, state $s\in\bbS$ and action $a^i\in\bbA^i$, we use $Q(i,s,a^i)$ to denote the corresponding column. Let $v^0$ be an arbitrary (starting) point in $\prod_{i=1}^n \Delta_{A^i}^S(d)$.

For each agent $i\in[n]$ and $s\in\bbS$, let $I_{i,s}=\{(i,s,a^i)|a^i\in\bbA^i\}$. Let $\mathcal{I}$ be a collection of all subsets $T$ of $\bigcup_{i\in[n],s\in\bbS}I_{i,s}$ such that for each $i$ and $s$ there is at least one element $(i,s,a^i)$ not in $T$.

For all $T\in\mathcal{I}$, we define $A(T)$, which is a subset of $\prod_{i=1}^n \Delta_{A^i}^S$, as follows.
$$A(T)=\left\{x\in\prod_{i=1}^n \Delta_{A^i}^S|x=v^0+\sum_{(i,s,a^i)\in T}\lambda(i,s,a^i)Q(i,s,a^i)\texttt{ for }\lambda(i,s,a^i)\geq 0 \right\}.$$

Let us fix some $T\in\mathcal{T}$, $\phi: [|T|]\rightarrow T$ be a permutation of $T$, and $w^0\in A(T)\cap \prod_{i=1}^n \Delta_{A^i}^S(d)$. We use $\Delta(w^0,\phi)$ to denote the convex hull of $|T|+1$ vertices $\{w^0,w^1,\cdots,w^{|T|}\}$ (which is a simplex), where $$w^i=w^{i-1}+Q(\phi(i)),\quad i\in[|T|].$$

Define $$\Sigma_T=\{\Delta(w^0,\phi)|\Delta(w^0,\phi)\in A(T)\cap \prod_{i=1}^n \Delta_{A^i}^S(d),\phi \texttt{ is a permutation of }T\}.$$

Then we have the following lemma.
\begin{lemma}[\cite{LaanT82computation}]
For each $T\in\mathcal{I}$, $\Sigma_{T}$ triangulates $A(T)$.
\end{lemma}

\textbf{The Vertices of {\sc End of the Line} Graph} are $\Sigma:=\bigcup_{T\in\mathcal{I}}\Sigma_T$.

\subsection{Proof of \Cref{lemma: stopping appro}}
\label{proof of stopping appro}

\approstopping*
\begin{proof}[Proof of \Cref{lemma: stopping appro}]
Because of the triangulation, we know that for any simplex $\delta\in\Sigma$ and two strategy profiles $\pi,\pi'\in\delta$, $\|\pi-\pi'\|_{\infty}\leq 1/d$.

Now let the simplex $\delta\in\Sigma$ be $(i,s)$-stopping. By the definition, we know for any $a^i\in\bbA^{i}$, there is a strategy profile, denoted by $\pi_{a^i}\in \prod_{i=1}^n\Delta_{A^i}^S$, whose label is $(i,s,a^i)$. Then $$\left(f(\pi_{a^i})\right)^{i}(s,a^i)-\pi_{a^i}^{i}(s,a^i)\leq 0 \quad \forall a^i\in\bbA^i.$$

Then for any $\pi\in\delta$, $\forall a^i\in\bbA^i$, we have $\pi_{a^i}^i(s,a^i)-\pi^i(s,a^i)\leq \frac{1}{d}$ and $f$ is $\lambda$-Lipschitz, which means $$\left(f(\pi)\right)^i(s,a^i)-\pi^i(s,a^i)\leq \left(f(\pi_{a^i})\right)^i(s,a^i)-\pi_{a^i}^i(s,a^i)+(\lambda+1)\frac{1}{d}\leq (\lambda+1)\frac{1}{d}.$$

Using $\sum_{b^i\in\bbA^i}\pi^i(s,b^i)=\sum_{b^i\in\bbA^i}\left(f(\pi^i)\right)(s,b^i)=1$, we have
\begin{eqnarray*}
	&& \left(f(\pi)\right)^i(s,a^i)-\pi^i(s,a^i)\\
	&=& \sum_{b^i\in\bbA^i,b^i\neq a^i}\pi^i(s,b^i)-\sum_{b^i\in\bbA^i,b^i\neq a^i}\left(f(\pi^i)\right)^i(s,b^i)\\
	&\geq& -(A_{\max}-1)(\lambda+1)\frac{1}{d}.
\end{eqnarray*}

Pick $a^i\in \bbA^i$ arbitrarily. Combine with the definition of labelling rule, for any $\pi\in\delta$, $j\in[n]$, $v\in\bbS$ and $b^j\in\bbA^j$, we have
\begin{eqnarray*}
	&& \left(f(\pi)\right)^j(v,b^j)-\pi^j(v,b^j)\\
	&\geq& \left(f(\pi_{a^i})\right)^j(v,b^j)-\pi_{a^i}^j(v,b^j)-(\lambda+1)\frac{1}{d}\\
	&\geq& \left(f(\pi_{a^i})\right)^i(s,a^i)-\pi_{a^i}^i(s,a^i)-(\lambda+1)\frac{1}{d}\\
	&\geq& -A_{\max}(\lambda+1)\frac{1}{d},
\end{eqnarray*}
which finishes the lower bound of this lemma.

Also, by the similar argument $\sum_{b^j\in\bbA^j}\pi^j(v,b^j)=\sum_{b^j\in\bbA^j}\left(f(\pi^j)\right)^j(v,b^j)=1$, we know that
\begin{eqnarray*}
	&& \left(f(\pi)\right)^j(v,b^j)-\pi^j(v,b^j)\\
	&\leq& A_{\max}(A_{\max}-1)(\lambda+1)\frac{1}{d}\\
	&\leq& A_{\max}^2(\lambda+1)\frac{1}{d},
\end{eqnarray*}
which finished the upper bound of this lemma.
\end{proof}

\subsection{Correctness of Our Choice of $d$}
\label{correct choice of d}

By \Cref{lemma: stopping appro}, we will find a $\pi$ such that $$\|f(\pi)-\pi\|_{\infty}\leq \dfrac{(1-\gamma)^5}{32A_{\max}^3R_{\max}^3}\frac{1}{L^2}.$$
In \Cref{lemma: single state}, we will have $$\epsilon'=\dfrac{(1-\gamma)^5}{32A_{\max}^3R_{\max}^3}\frac{1}{L^2}\left(1+\frac{A_{\max}R_{\max}}{1-\gamma}\right)\leq \dfrac{(1-\gamma)^5}{32A_{\max}^3R_{\max}^3}\frac{1}{L^2}\frac{2A_{\max}R_{\max}}{1-\gamma}\leq \dfrac{(1-\gamma)^4}{16A_{\max}^2R_{\max}^2}\frac{1}{L^2},$$
which means 
\begin{eqnarray*}
	&&\max\left(0,V^{\pi^i,\pi^{-i}}_{\pi^i(s,a^i)=1}(s)-V^{\pi^i,\pi^{-i}}(s)\right)\\
	&\leq& A_{\max}\left(\frac{\sqrt{\epsilon'}}{1-\gamma}+R_{\max}\sqrt{\epsilon'}+\epsilon'\right)\\
	&\leq& A_{\max}\left(\frac{2R_{\max}}{1-\gamma}\sqrt{\epsilon'}+\epsilon'\right)\\
	&\leq& A_{\max}\left(\frac{2R_{\max}}{1-\gamma}\dfrac{(1-\gamma)^2}{4A_{\max}R_{\max}}\frac{1}{L}+\dfrac{(1-\gamma)^4}{16A_{\max}^2R_{\max}^2}\frac{1}{L^2}\right)\\
	&\leq& \dfrac{(1-\gamma)}{2}\frac{1}{L}+\dfrac{(1-\gamma)^4}{16A_{\max}R_{\max}^2}\frac{1}{L^2}.
\end{eqnarray*}
By \Cref{lemma: all states}, we know $\pi$ is an $\frac{1}{2L}+\frac{(1-\gamma)^3}{16A_{\max}R_{\max}^2}\frac{1}{L^2}$-approximate MPE, so $\pi$ is a $1/L$-approximate MPE.

\end{document}